%% file: Maximum_MST_Ratio_LNCS.tex
\documentclass[runningheads]{llncs}
\usepackage[T1]{fontenc}
\usepackage{graphicx}
\usepackage{color}
\usepackage{amsmath}
\usepackage{amssymb}
\usepackage{xspace}
\usepackage{epsfig}
\usepackage[dvipsnames]{xcolor}
\usepackage{soul}
\usepackage{xfrac}

\usepackage[colorinlistoftodos,prependcaption,textsize=tiny]{todonotes}

\usepackage[normalem]{ulem}
\newcommand {\mm}[1] {\ifmmode{#1}\else{\mbox{\(#1\)}}\fi}

\newcommand{\ignore}[1]{}

\newsavebox{\smallProofsym}                 
\savebox{\smallProofsym}                               %
{
\begin{picture}(6,6)
\put(0,0){\framebox(6,6){}}
\put(0,2){\framebox(4,4){}}
\end{picture} 
}  

\newcommand{\Rspace}        {\mm{{\mathbb R}}}

\newcommand{\Oh}            {\mm{O}}

\newcommand{\ee}            {\mm{\varepsilon}}

\newcommand{\qedhere}{\hfill$\qed$}

\newcommand{\Skip}[1]       {}

\begin{document}
\title{On the MST-ratio: Theoretical Bounds and Complexity of Finding the Maximum}
\titlerunning{On the MST-ratio}
%

\author{Afrouz Jabal Ameli\inst{1}\orcidID{0000-0002-9823-683} \and
Faezeh Motiei\inst{1}\orcidID{0009-0006-8805-9307} \and
Morteza Saghafian\inst{2}\orcidID{0000-0002-4201-5775}}
\authorrunning{A. Jabal Ameli, F. Motiei, M. Saghafian}
%
\institute{Eindhoven University of Technology \\
5600MB Eindhoven, the Netherlands\\
\email{\{a.jabal.ameli,f.motiei\}@tue.nl}\\
\url{www.tue.nl} \and
IST Austria (Institute of Science and Technology Austria) \\ Kloster\-neu\-burg, Austria\\
\email{morteza.saghafian@ist.ac.at}\\
\url{www.ist.ac.at}
}
\maketitle              
\begin{abstract}
Given a finite set of red and blue points in $\Rspace^d$, the MST-ratio is defined as the total length of the Euclidean minimum spanning trees of the red points and the blue points, divided by the length of the Euclidean minimum spanning tree of their union. The MST-ratio has recently gained attention due to its direct interpretation in topological models for studying point sets with applications in spatial biology. The maximum MST-ratio of a point set is the maximum MST-ratio over all proper colorings of its points by red and blue. We prove that finding the maximum MST-ratio of a given point set is NP-hard when the dimension is part of the input.  Moreover, we present a quadratic-time $3$-approximation algorithm for this problem. As part of the proof, we show that, in any metric space, the maximum MST-ratio is smaller than $3$. Additionally, we study the average MST-ratio over all colorings of a set of $n$ points. We show that this average is always at least $\frac{n-2}{n-1}$, and for $n$ random points uniformly distributed in a $d$-dimensional unit cube, the average tends to $\sqrt[d]{2}$ in expectation as $n$ approaches infinity.

\end{abstract}
\section{Introduction}\label{sec:1}
    Recently, motivated by applications in spatial biology, Cultrera et al. \cite{CDES22} studied the interactions between color classes in a colored point set from a topological point of view. To this end, they developed a framework based on the \textit{chromatic Delaunay mosaic} and explored its combinatorial and topological properties; see also \cite{BCDES22}. Moreover, they introduced the concept of \textit{MST-ratio} as one of the measures for the mingling of points with different colors in a colored point set. Investigating this measure opens the door to interesting discrete geometry questions related to Euclidean minimum spanning trees (EMST), which is the main focus of this paper. For a set $P$ of $n$ points in the Euclidean space $\Rspace^d$, an \textit{Euclidean minimum spanning tree} of $P$, denoted by $\text{EMST}(P)$, is a minimum spanning tree (MST) of the complete geometric graph on $P$ where the weight of each edge is the Euclidean distance between its endpoints. For a partition $P=R \cup B$ of $P$ into red and blue points, the EMST-ratio (sometimes referred to as the MST-ratio) can be defined as follows.
    \begin{equation*}
    \mu(P,R) = \frac{|\text{EMST}(R)|+|\text{EMST}(P\setminus R)|}{|\text{EMST}(P)|}.
    \end{equation*}

where $|\text{EMST}(X)|$ for a point set $X$ is the length of an Euclidean minimum spanning tree of $X$. How much longer can EMSTs of two finite sets be compared to an EMST of their union? Given a point set $P$, we are
interested in the maximum ratio, $\gamma(P)$, over all proper partitionings of $P$ into two sets. The upper and lower bounds for $\gamma$ of certain classes of point sets can be found in \cite{CDES24,DPT23}, see Section~\ref{subsec:related} for further details. However, limited research has been conducted on the maximum EMST-ratio in higher-dimensional Euclidean spaces.
Moreover, the question of whether there is an efficient algorithm to compute $\gamma$ for a given point set remains unanswered.

The question can also be raised in the abstract setting. Namely, given a weighted complete graph $G$ with positive weights and a bi-partition $V(G)=R\cup B$ of the vertices of $G$, we define the MST-ratio of this partition as:
\begin{equation*}
    \mu(G,R) = \frac{|\text{MST}(G[R])|+|\text{MST}(G[P\setminus R])|}{|\text{MST}(G)|},
\end{equation*}
where $G[R]$ and $G[B]$ are the induced subgraphs of $G$ on the vertex sets $R$ and $B$, respectively, and $|H|$ is the total weight of the edges of the graph $H$. The maximum MST-ratio of $G$, denoted by $\gamma(G)$, is the maximum ratio over all proper bi-partitions of $V(G)$. 

In the \textit{MAX-MST-ratio} problem, we aim to find a partition $V(G)=R\cup B$~of the vertices of $G$ that maximizes the MST-ratio. As for the geometric setting, we introduce the \textit{MAX-EMST-ratio} problem, in which given dimension $d$ and a point set $P$ in $\Rspace^d$, we aim to find the maximum EMST-ratio of $P$. When $d$ is a fixed number, we call the problem \textit{d-MAX-EMST-ratio}.

Our work establishes new bounds on the maximum and average MST-ratio. We also analyze the computational complexity of MAX-MST-ratio and study~it through the lens of approximation algorithms. By an $\alpha$-approximation ($\alpha~\ge~1$) for MAX-MST-ratio, we mean a polynomial-time algorithm that for every weighted graph $G$ provides a bi-partition $V(G)=R\cup B$ satisfying $\alpha \cdot \mu(G,R)\ge \gamma(G)$.

We summarize our main results as follows:

\begin{itemize}
    \item The MAX-MST-ratio problem is NP-hard and hard to approximate even within a factor of $O(n^{1-\varepsilon})$ (Theorems~\ref{thm:graphsuperhard}).
    \item The MAX-EMST-ratio problem is NP-hard (Theorem~\ref{thm:emstnphard}).
    \item In any metric space, the maximum MST-ratio is smaller than $3$ (Theorem~\ref{thm:metricupperbound}).
    \item There exists an $\Oh(n^2)$ time $3$-approximation algorithm for the MAX-EMST-ratio problem (Theorem~\ref{thm:emstapprox}).
    \item For every set of $n$ points in Euclidean space, the average EMST-ratio over all colorings is at least $\frac{n-2}{n-1}$ (Theorem~\ref{thm:averagemst}).
    \item The average EMST-ratio for a set of $n$ random points uniformly distributed in $[0,1]^d$ tends to $\sqrt[d]{2}$ as $n$ goes to infinity (Theorem~\ref{thm:randomaverage}).
\end{itemize}

\subsection{Related Work}\label{subsec:related}
The maximum EMST-ratio is closely related to the Steiner ratio of the Euclidean space. Recall that the Steiner tree of a point set is a tree that connects the points via segments of minimum total length, allowing using some extra points. The Steiner ratio, $\rho_d$ of $\Rspace^d$, is the infimum over all finite point sets in $\Rspace^d$ of the length of the Steiner tree divided by the length of an EMST of the set. Figure~\ref{fig:Steinerexample} shows an example of three points that form an equilateral triangle of side length $1$ with Steiner ratio $\sqrt{3}/2$. Gilbert and Pollack \cite{GP68} conjectured that this is the most extreme example in the plane and therefore $\rho_2 = \sqrt{3}/2 \approx 0.866$. The best-known lower bound for $\rho_2$ is $0.824...$, due to Chung and Graham \cite{CG85}. It is also important for us to have a universal lower bound on $\rho_d$ that holds for all $d$. Gilbert and Pollack \cite{GP68} presented a short proof for $\rho_d \geq 0.5$. Later, Graham and Hwang \cite{GH76} showed $\rho_d \geq 0.577$ for any $d$.
\begin{figure}[ht]
    \centering
    \input{./Figures/SteinerRatioMax.tex}
    \caption{The point set $P$ shown by square nodes is $3$ vertices of an equilateral triangle of side length $1$. The black edges form $\text{EMST}(P)$ with length $2$. Adding an extra vertex in the center, the green edges form the minimum Steiner tree of $P$ with length $\sqrt{3}$.}
    \label{fig:Steinerexample}
\end{figure}
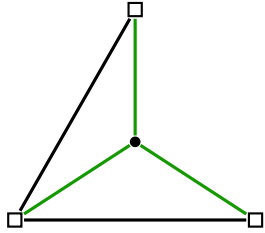

Using the above bounds for $\rho_2$, Cultrera et al. \cite{CDES24} showed that the supremum over all point sets in $\Rspace^2$ of the maximum EMST-ratio is between $2.154$ and $2.427$. It is not hard to see that the infimum of the maximum EMST-ratio is $1$. Their attention then shifted to lattice point sets, where they introduced the MST-ratio (despite the infinite number of points). They subsequently demonstrated that the infimum and supremum of the maximum MST-ratios across all $2$-dimensional lattices are $1.25$ and $2$, respectively. 

Dumitrescu et al. \cite{DPT23} proved that for any set $P$ of at least $12$ points in $\Rspace^2$, it holds that $\gamma(P) > 1$. They also showed that for $n$ points sampled uniformly at random in $[0,1]^2$, the maximum EMST-ratio is at least $\sqrt{2} - \ee$, for every $\ee > 0$, with a probability that tends to $1$ as $n$ goes to infinity. A tighter analysis shows that in the $d$-dimensional unit cube, the expected value of the average EMST-ratio over all colorings tends to $\sqrt[d]{2}$ as $n$ goes to infinity, see Theorem~\ref{thm:randomaverage}. The proof relies on the classic result by Beardwood, Halton, and Hammersley \cite{BHH59}, saying that the length of an MST of a set $P$ of $n$ points uniformly distributed in $[0,1]^d$ (or any bounded region of volume $1$ in $\Rspace^d$) satisfies
\begin{align}
\label{equ:beta}
   |\text{EMST}(P)|/n^{1-1/d} \rightarrow{} \beta(d)
\end{align}

with probability $1$, where $\beta(d) > 0$ is a constant depending only on the dimension. Best known lower and upper bounds for $\beta(2)$ are approximately $0.6$ by Avram and Bertsimas \cite{AvBe92} and $0.707$ by Gilbert \cite{Gil65}, respectively.


\subsection{Outline}

Section~\ref{sec:abstract} focuses on abstract graphs and presents complexity results on the MAX-MST-ratio problem. Section~\ref{sec:geometry} addresses the geometric setting, establishing bounds for the maximum MST-ratio and providing complexity results for MAX-EMST-ratio problem. Section~\ref{sec:average} explores the average EMST-ratio across all colorings of a point set. Section~\ref{sec:discussion} concludes with open questions to deepen understanding of the MST-ratio. Section~\ref{sec:experiments} presents experimental results regarding the maximum, average, and bipartite EMST-ratio for random point sets.

\section{Maximum MST-ratio of Abstract Graphs}\label{sec:abstract}
In this section, we study the complexity of finding the maximum MST-ratio for abstract graphs. Throughout this section and the subsequent sections, we call an edge \textit{colorful} if the two endpoints of it have different colors; otherwise, we call it \textit{monochromatic}. 

As the main result of this section, we show that not only MAX-MST-ratio is NP-hard, but it is extremely hard to approximate. In Section~\ref{sec:geometry}, we will use the NP-hardness of the MAX-MST-ratio to demonstrate that the geometric version of the problem (i.e., the MAX-EMST-ratio) is also NP-hard. Theorem~\ref{thm:graphsuperhard} presents a strong inapproximability result for the MAX-MST-ratio problem.



\begin{theorem}\label{thm:graphsuperhard}
    For every $0<\varepsilon\le 1$, there is no polynomial-time $O(n^{1-\varepsilon})$-approximation with weights restricted to $1$ and $n$ for MAX-MST-ratio problem, unless P=NP.
\end{theorem}

\begin{proof}
    We prove that given an $\alpha$-approximation for MAX-MST-ratio, one can obtain a $2\alpha$-approximation for the MAX-Clique problem.
    Let $G=(V,E)$ be a graph on $n$ vertices on which we wish to find the maximum clique. From this, we create an instance of MAX-MST-ratio with the input graph $G'$, where $G'$ is a complete graph with $V(G')=V$ and the weight of any edge $e \in E(G')$ is $n$ if $e\in E$, and it is $1$ otherwise. Given a bi-partition of vertices into sets $R$ and $B$, we refer to $|\text{MST}(R)|+|\text{MST}(B)|$ as the \textit{value} of this bi-partition.

    Given a clique $C$ of size $\ell$ in $G$, we can find a coloring of value $k$
     for $G'$ in polynomial time, where $\ell\le \lfloor \frac{k}{n}\rfloor+2$. This is trivial if $\ell\le 2$, and for $\ell\ge 3$ it is sufficient to color all but one of the vertices in $C$ by red, and the remaining vertices by blue. In this case, the weight of any MST of the red vertices is $(\ell -2) \cdot n$ 
     and thus $\ell\le \lfloor \frac{k}{n}\rfloor+2$.
     
    Conversely, given a coloring $Q$ of value $k$ for $G'$, we show that we can find a clique of size at least $\frac{\lfloor k/n\rfloor}{2}+1$ in $G$ in polynomial time. Let $R_Q$ and $B_Q$ denote a red MST and a blue MST of $G$, respectively, in the coloring $Q$. Given that the weights in $G'$ belong to $\{1,n\}$ and since $R_Q \cup B_Q$ has exactly $n-2$ edges, then $R_Q \cup B_Q$ has exactly $\lfloor k/n\rfloor$ edges of weight $n$.
    W.l.o.g, we can assume that $R_Q$ has at least $ \frac{\lfloor k/n\rfloor}{2}$ edges of weight $n$. Consider $F:=(V(R_Q),E'')$, where $E''$ is the set of edges of weight $1$ in $R_Q$. Let $C_1,\ldots,C_r$ be the connected components of $F$. Clearly, we have $r\ge \frac{\lfloor k/n\rfloor}{2}+1$. For every $1\le i\le r$, let $u_i$ be a vertex in $C_i$. As $R_Q$ is an MST on the set of red vertices, then for every $i<j$, the weight of the edge $u_iu_j$ in $G$ is $n$. Hence, $\{u_1,u_2,\ldots,u_r\}$ is a clique of size $r\ge \frac{\lfloor k/n\rfloor}{2}+1$~in~$G$.
    
    Altogether, this shows that one can use an $\alpha$-approximation algorithm for the MAX-MST-ratio problem to find a $2\alpha$-approximation for the MAX-Clique problem. Zuckerman~\cite{Z06} proved that there exists no  $O(n^{1-\varepsilon})$-approximation for MAX-Clique problem, unless $P=NP$. Therefore, there exists no $O(n^{1-\varepsilon})$-approximation for MAX-MST-ratio, unless $P=NP$.
    \qedhere
\end{proof}

\section{Maximum MST-ratio of Geometric Point Sets}\label{sec:geometry}
We show that MAX-EMST-ratio problem is NP-hard. We remark that, a weighted complete graph does not necessarily correspond to the distance graph of a point set in $\mathbb{R}^d$. A weighted complete graph $G$ on $n$ vertices is \textit{realizable} in $\mathbb{R}^d$ if there exist $n$ points in $\mathbb{R}^d$ such that the distance graph of these points is isomorphic to $G$.

We rely on a result by Dekster and Wilker \cite{DW87}, which shows that any weighted complete graph, whose weights belong to a small range, is realizable in Euclidean space.


    



\begin{lemma}\cite[Theorem 2]{DW87}\label{lem:nearregular}
    For every positive integer $n>0$ there exists a real number $ \lambda_n = \sqrt{1-\frac{1}{\Theta(n)}}$ such that every complete graph $G$ on $n$ vertices whose all edges have weights in $[\lambda_n \cdot \ell,\ell]$, for a positive real number $\ell$ is realizable in $\Rspace^{n-1}$. 
\end{lemma}

We use the above lemma to increase all the weights of an instance in Theorem~\ref{thm:graphsuperhard} by a fixed number in order to make it realizable in Euclidean space.

\begin{lemma}\label{lem:increase}
    Let $G$ be a weighted complete graph $G$ on $n$ vertices. Let $w_1$ and $w_2$ be the minimum and maximum weights of the edges of $G$ and let $\lambda_n$ be as defined in Lemma~\ref{lem:nearregular}.
For every $N$ such that $N\ge \max\{|w_1|,\frac{\lambda_nw_2-w_1}{1-\lambda_n}\}$, by increasing the weight of all edges in $G$ by $N$, we obtain a realizable graph in $\Rspace^{n-1}$.
\end{lemma}
\begin{proof}
     Increasing the weights of the edges in $G$ by $N$ results in a graph $G'$ with positive weights in which the minimum and maximum weights are $N+w_1$ and $N+w_2$. As $N> \frac{\lambda_nw_2-w_1}{1-\lambda_n}$ then $N+w_1 \geq \lambda_n(N+w_2)$ and hence by Lemma~\ref{lem:nearregular}, $G'$ is realizable in $\Rspace^{n-1}$.
     \qedhere
\end{proof}
\begin{theorem}\label{thm:emstnphard}
    The MAX-EMST-ratio problem is NP-hard.
\end{theorem}
\begin{proof}
    By Lemma~\ref{lem:increase}, the instances of the MAX-MST-ratio problem with weights $1$ and $n$ can be turned into a geometric graph by increasing the weights of all edges by a fixed amount. Note that this action does not change the optimal coloring of the graph in terms of the maximum MST-ratio. This is because the total number of edges in the MSTs of the two colors is always $n-2$. Moreover, the amount we add is $O(n^2)$, according to Lemma~\ref{lem:increase}. Thus, from any instance of the MAX-MST-ratio with weights $1$ and $n$, we can construct an instance of the MAX-EMST-ratio problem in polynomial time and hence by Theorem~\ref{thm:graphsuperhard}.
    \qedhere
\end{proof}

As established in Theorem~\ref{thm:emstnphard}, determining the maximum EMST-ratio is NP-hard. We conjecture that the problem remains NP-hard even when restricted to a specific dimension. In particular, we are interested in $\Rspace^2$ and $\Rspace^3$, as they are relevant for practical applications.
\begin{conjecture}
    The $k$-MAX-EMST-ratio problem for $k \in \{2,3\}$ is NP-hard.
\end{conjecture}

Theorem~\ref{thm:graphsuperhard} shows that it is even hard to approximate the maximum MST-ratio for general graphs. However, in Euclidean space, it is not difficult to approximate the maximum EMST-ratio as it cannot be too large. We show a general upper bound for the maximum MST-ratio in any metric space that leads to a simple $3$-approximation algorithm for the MAX-EMST-ratio problem. In the proof, we use the following well-known fact about the path double cover of trees.

\begin{lemma}[\cite{PDC14}]\label{lem:PartitionToPaths}
    Given a tree $T$ with $\ell$ leaves, there exist $\ell$ (not necessarily distinct) paths in $T$ such that each edge of $T$ is in precisely two of these paths.
\end{lemma}


\begin{theorem}\label{thm:metricupperbound}
    In any metric space, the maximum MST-ratio for any set of $n \ge 5$ points is at most $3-\frac{4}{n-1}$.
\end{theorem}
\begin{proof}
    Consider an arbitrary point set $S$ in a metric space, and let $T$ be an MST of $S$. 
    Let $P^*$ be a path with maximum total weight among all paths in $T$ whose endpoints are not leaves of $T$. 
    Assume w.l.o.g. that one endpoint of $P^*$, denoted by $r$, is red, and we hang $T$ from $r$.
    We traverse the tree from the root $r$ using Depth-First Search (DFS), prioritizing the vertices of $P^*$ when exploring branches. The starting time of a node refers to the moment it is first visited during the traversal.
    
    To prove the statement, for every coloring $Q$ of $S$ we build spanning trees $R_Q$ and $B_Q$ on the set of red vertices and the set of blue vertices of $Q$ respectively, such that $|R_Q \cup B_Q|$ is at most $(3-\frac{4}{n-1})|T|$. We construct $R_Q$ and $B_Q$ in the following way (see Figure~\ref{fig:DFS}).
    \begin{itemize}
        \item We add every edge $e$ of $T$ with both endpoints in red (blue) to $R_Q$ ($B_Q$).
        \item Consider any red (blue) vertex $v$ in $T$ and any blue (red) child $w$ of $v$ such that there exists at least one red (blue) vertex in the descendants of $w$ (e.g. $v=r$ in Figure~\ref{fig:DFS}). Let $u_1, u_2, \ldots, u_x$ be the red descendants of $w$ for which all the internal vertices of the $wu_i$-path in $T$ are blue (red). Assume that these vertices are sorted by starting time (i.e. $u_1$ has the earliest starting~time, and $u_x$ has the latest).
        We add the edges $u_1u_2,u_2u_3,\ldots,u_{x-1}u_x,u_xv$ to $R_Q$ (to $B_Q$).
        \item Let $v_1$, $v_2$, $\ldots$, $v_y$ be the set of blue vertices that have no blue ancestors, sorted by starting time. We add the edges $v_1v_2,v_2v_3,\ldots,v_{y-1}v_y$ to $B_Q$.
    \end{itemize}

    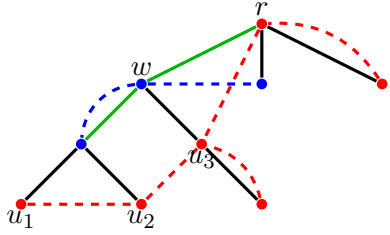
\begin{figure}[ht]
    \centering
        \input{./Figures/DFS.tex}
        \caption{The construction of $R_Q$ and $B_Q$ (dashed edges) according to the coloring $Q$ of vertices in the proof of Theorem~\ref{thm:metricupperbound}. The solid black and green edges form the tree $T$ rooted at $r$. The children of any vertex are from left to right according to the order of their starting time in DFS. 
        The green edges are the edges of $P^*$.
        }
    \label{fig:DFS}
\end{figure}
Let $\mathcal{P}$ be the set of all $uv$-paths $P_e$ in $T$, where $e=uv$ is an edge in $R_Q\cup B_Q$.
    Observe that every edge of $T$ belongs to at most $3$ paths in $\mathcal{P}$. Namely, by construction, if $x$ is the parent of $y$ and $x$ is blue (similarly red) then the edge $xy$ appears in at most one path $P_{e_b}$ where $e_b\in B_Q$ and at most two paths $P_{e_r}$ where $e_r\in R_Q$.
    Furthermore, certain edges in $T$ are covered by at most two paths in $\mathcal{P}$, namely the leaf edges (i.e. edges incident to a leaf) and the edges in $P^*$ (i.e. the green path in Figure~\ref{fig:DFS}).
    
    Since the instance is metric, the weight of every edge $uv$ in $R_Q \cup B_Q$ is not more than the total weight of the corresponding $uv$-path in $\mathcal{P}$. So we have $|R_Q|+|B_Q|\le 3|T|-W_L-W_{P^*}$, where $L$ is the set of leaf edges in $T$ and $W_L$ and $W_{P^*}$ are the total weight of the edges of $L$ and $P^*$, respectively. 
    It suffices to show that $W_{P^*}+W_L\ge (\frac{4}{n-1})|T|$.
            
    If the diameter of $T$ is less than $4$, then $P^*\cup L=T$, and hence $W_{P^*}+W_L=|T|\ge (\frac{4}{n-1})|T|$.
    Assume that the diameter of $T$ is at least $4$. Let $T'$ be the tree obtained from $T$ by removing all the leaves of $T$ and let $\ell$ and $\ell'$ be the number of leaves of $T$ and $T'$, respectively. Note that a vertex in $T$ that is a leaf in $T'$ must be adjacent to at least one leaf in $T$ and hence $\ell'\le \ell$. Furthermore, as the diameter of $T$ is at least $4$, then $T'$ has at least one non-leaf vertex which implies that $\ell+\ell'\le n-1$ and thus, $\ell'\le \frac{n-1}{2}$. 
    Using Lemma~\ref{lem:PartitionToPaths}, there are $\ell'$ paths in $T'$ such that all edges of $T'$ belong to exactly two such paths. Hence, $W_{P^*}$ is lower bounded by $\frac{2|T'|}{\ell'}=\frac{2(|T|-W_L)}{\ell'}\ge \frac{4(|T|-W_L)}{n-1}$.
    Therefore, it can be concluded that $W_{P^*}+W_L\ge \frac{4(|T|-W_L)}{n-1}+W_L\ge \frac{4}{n-1}|T|.$  
    \qedhere
\end{proof}

Note that the above bound is the best we can achieve in metric spaces. In fact,
    for every odd $n$, there exist metric spaces with $n$ points (elements) for which the maximum MST-ratio is arbitrarily close to 
    $3-\frac{4}{n-1}$. To see this,
    let $n=2k+1$, and define the point set $S=\{v_1,\ldots,v_{n}\}$. For every $i<n$, we set the weight of $v_{n}v_i$ to $1$. For every $1 \le i\le k$, we set the weight of $v_{2i}v_{2i-1}$ to be $\varepsilon$ for a sufficiently small $\varepsilon>0$. Any other edge has weight $2$.

    We first argue that the weights are metric. The edges of weight $1$ are the edges incident to $v_{n}$, and the edges of weight $\varepsilon$ form a matching. Therefore, any triangle with edge weights $a\le b \le c$ has zero or two edges of weight $1$ and has at most one edge of weight $\varepsilon$. Thus, $(a,b,c)$ is in $\{(\varepsilon,1,1),(\varepsilon,2,2),(1,1,2),(2,2,2)\}$, and in all cases we have $c\le a+b$.
    Observe that $|\text{MST}(S)|=\varepsilon\times k + 1 \times k=k(1+\varepsilon)$.
    If we color $v_1,v_3,\ldots, v_{n}$ by red and $v_2,v_4,\ldots, v_{n-1}$ by blue, then $|\text{MST}(R)|=k$ and $|\text{MST}(B)|=2(k-1)=2k-2$. Thus, $|\text{MST}(R)|+|\text{MST}(B)|=3k-2$ and hence the MST-ratio is $\frac{3k-2}{k(1+\varepsilon)}=(3-\frac{4}{n-1})\frac{1}{1+\varepsilon}$.
    \qedhere


\begin{theorem}\label{thm:emstapprox}
There exists an $\Oh(n^2)$-time $3$-approximation algorithm for the MAX-EMST-ratio problem.    
\end{theorem}
\begin{proof}
    By Theorem~\ref{thm:metricupperbound}, the maximum EMST-ratio for an $n$-element point set is upper bounded  by $3-\frac{4}{n-1}$.
    On the other hand, one can achieve an EMST-ratio of at least $\frac{n-2}{n-1}$ by first computing a minimum spanning tree of the point set in $\Oh(n^2)$ (see \cite{CLRS,Prim57}), and then coloring the vertices of the two components of it after removing the shortest edge by red and blue. Finally, as $\frac{n-2}{n-1} \geq \frac{(3-\frac{4}{n-1})}{3}$, this is an $\Oh(n^2)$ time $3$-approximation algorithm.  
    \qedhere
\end{proof}

We remark that a slightly improved running time for computing an EMST is provided in \cite{AESW07}. Moreover, one could achieve a weaker approximation factor $3.47$ in a simpler way using a trivial upper bound $\frac{2}{\rho_d}\le \frac{2}{0.577}< 3.47$ on the MAX-EMST-ratio.
Indeed, if we restrict to $\Rspace^2$ (i.e. $2$-MAX-EMST-ratio), the approximation factor in Theorem~\ref{thm:emstapprox} can be improved to $2.427$, because the upper bound for the maximum EMST-ratio in the plane is $ \approx 2.427$ (see \cite{CDES24}), and for $n>12$ a coloring with EMST-ratio greater than $1$ can be computed in polynomial time (see \cite{DPT23}).
As can be observed, the coloring used to present the $3$-approximation is fairly straightforward. However, there seems to be room for improvement by identifying more mingled colorings. 

We introduce \textit{Bipartite Colorings} as good candidates to approximate the maximum EMST-ratio of a point set. By a bipartite coloring of a point set, we mean considering one of the Euclidean minimum spanning trees and then partitioning the point set into two subsets labeled by colors in which every edge of the taken EMST is colorful. Note that since the EMST of a point set is not necessarily unique, it may have more than one bipartite coloring, namely, one for each EMST. Intuitively, bipartite colorings give well-distributed colorings of the point set that can be computed in $\Oh(n^2)$ time ($O(n\log n)$ in the plane, see \cite{BCKO08}). Our experiments suggest that bipartite colorings give a very good approximation of the maximum EMST-ratio for random point sets in the plane (see Section~\ref{sec:experiments}). 

Although bipartite colorings seem to be effective, they do not always produce the maximum EMST-ratio. Here is an extreme example:
Consider $n=2k+1$ points on a zigzag path in the triangular grid and call them a \textit{Triangular Chain} (see Figure~\ref{fig:BipartitevsMax}). By stretching the chain slightly, one can ensure a unique EMST and hence a unique bipartite coloring. The maximum EMST-ratio of a triangular chain with $2k+1$ points is $\frac{3k-2}{2k}$ (The proof is provided in Section~\ref{subsec:zigzagproof}), while the ratio for a bipartite coloring is $\frac{2k-1}{2k}$ . As $k$ increases, the proportion between the maximum EMST-ratio and this bipartite EMST-ratio approaches $1.5$. 

\begin{figure}[ht]
    \centering
    \includegraphics[width=0.9\linewidth]{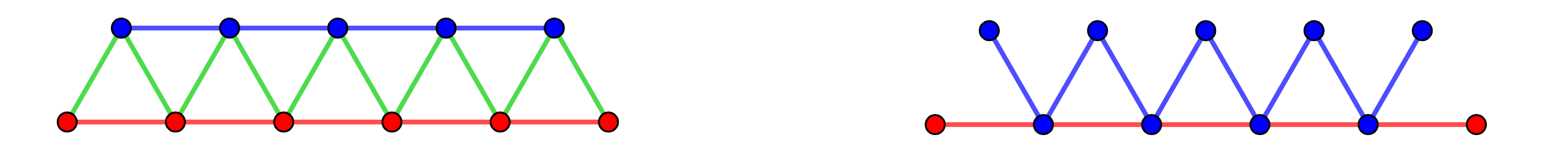}
    \caption{An example of a point set where the maximum EMST-ratio has a significant gap from the EMST-ratio of a bipartite coloring. The left panel shows a triangular chain for $k=5$, its EMST in green, and the red and blue EMSTs in a bipartite coloring.
    The right panel presents the coloring that achieves the maximum EMST-ratio.}
    \label{fig:BipartitevsMax}
\end{figure}
We conjecture that this is the worst approximation factor of the maximum EMST-ratio that can be obtained from a bipartite coloring for a point set in the plane. If this conjecture is correct, it would lead to an $O(n\log n)$ time algorithm that achieves a $(1.5)$-approximation for the $2$-MAX-EMST-ratio problem.

\begin{conjecture} For every point set in $\Rspace^2$, the EMST-ratio for any bipartite coloring is a $(1.5)$-approximation for the maximum EMST-ratio.
\end{conjecture}

Apart from the bipartite colorings, we are interested in any improved approximation algorithm for the MAX-EMST-ratio that can be obtained by a simple mixed coloring of a point set.

\begin{question}
    Does a polynomial-time $2$-approximation algorithm exist for MAX-EMST-ratio?
\end{question}

\section{Average MST-ratio.}\label{sec:average}
The \textit{Average EMST-ratio} of a point set $P$ is the average of EMST-ratio over all proper $2$-colorings of $P$. It is automatically a lower bound for the maximum EMST-ratio. Besides, it could be considered as a measure to see how close a point set is to a typical random point set, as we know the behavior of random point sets. In fact, it is proved that the expected value for the average EMST-ratio of a random point set in a unit square is $\sqrt{2}$ in the limit \cite{DS24}. We present the proof for the generalized version of this result, i.e. for points within a $d$-dimensional unit cube (or any bounded region in $\Rspace^d$).

\begin{theorem}\label{thm:randomaverage}
    For a set of $n$ random points uniformly distributed in $[0,1]^d$ the expected value of the average EMST-ratio tends to $\sqrt[d]{2}$ as $n$ goes to infinity. 
\end{theorem}
\begin{proof}
    Let $n$ be large enough and $P$ be a set of $n$ random points uniformly distributed in $[0,1]^d$. 
    Consider all $2^n -2$ proper bi-partitions $P = R \cup B$ into red and blue points. 
    Clearly, for $1 \leq k \leq n-1$, in $\binom{n}{k}$ of such bi-partitions $|R|=k, |B|=n-k$. Ignoring the relatively small $k$ or $n-k$, by (\ref{equ:beta}), we know that for every $k_0 \leq k \leq n-k_0$ we have 
    \begin{align*}
    |\text{EMST}(R)| \sim &\beta(d).k^{1-1/d}, \hspace{2mm} |\text{EMST}(B)| \sim \beta(d).(n-k)^{1-1/d},\\ &|\text{EMST}(P)| \sim \beta(d).n^{1-1/d}
    \end{align*}
    with probability approaching $1$, where $\beta(d)$ is the constant in (\ref{equ:beta}), and $k_0$ is an appropriate constant.
    Since $n$ is significantly greater than $k_0$, and the EMST-ratio is at most $3$, the contribution of the EMST-ratios for colorings with $k < k_0$ red points or with $n-k < k_0$ blue points are negligible when computing the average EMST-ratio. Thus, the average EMST-ratio over all proper bi-partitions of $P$ is
    \begin{align*}
    \sum_{P=R \cup B} \frac{|\text{EMST}(R)|+|\text{EMST}(B)|}{|\text{EMST}(P)|} &\sim \sum_{k=k_0}^{n-k_0} \frac{[\beta(d).k^{1-1/d} + \beta(d).(n-k)^{1-1/d}].\binom{n}{k}}{\beta(d).n^{1-1/d}.(2^n-2)}\\
    &= \sum_{k=k_0}^{n-k_0} \frac{[k^{1-1/d} + (n-k)^{1-1/d}].\binom{n}{k}}{n^{1-1/d}.(2^n-2)}
    \end{align*}
    with probability $1$ as $n$ goes to infinity. The right-hand side is a special case of convergence of Bernstein polynomials (see \cite{Ber12,KS07}), and therefore converges to $\sqrt[d]{2}$ as $n$ goes to infinity.
    \qedhere
\end{proof}

\begin{figure}[ht]
    \centering
        \begin{tikzpicture}
            \node at (0,0){\includegraphics[width=0.45\linewidth]{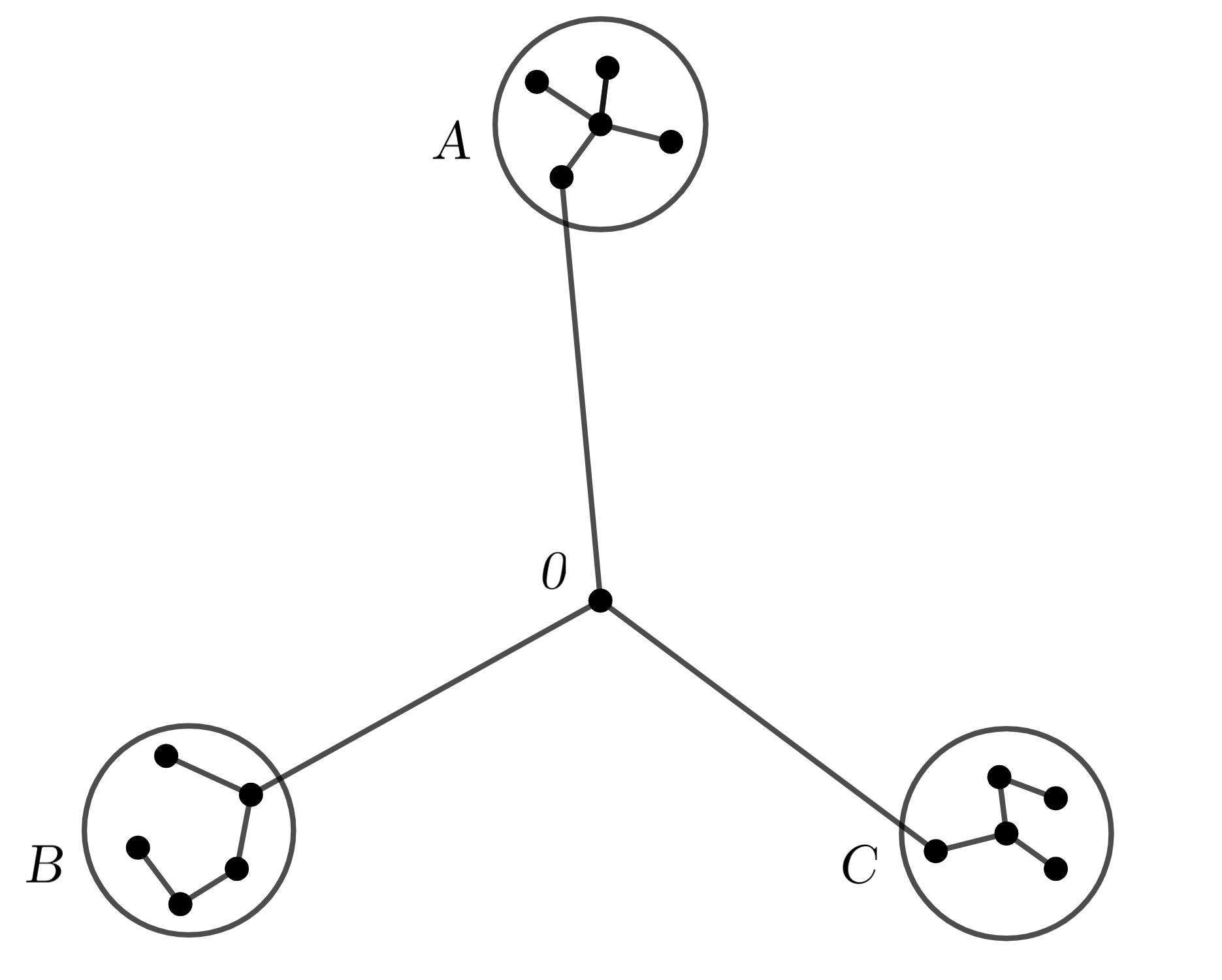}};
            \node at (-0.2,-2.6){$(a)$};
            \node at (6,0.2){\includegraphics[width=0.4\linewidth]{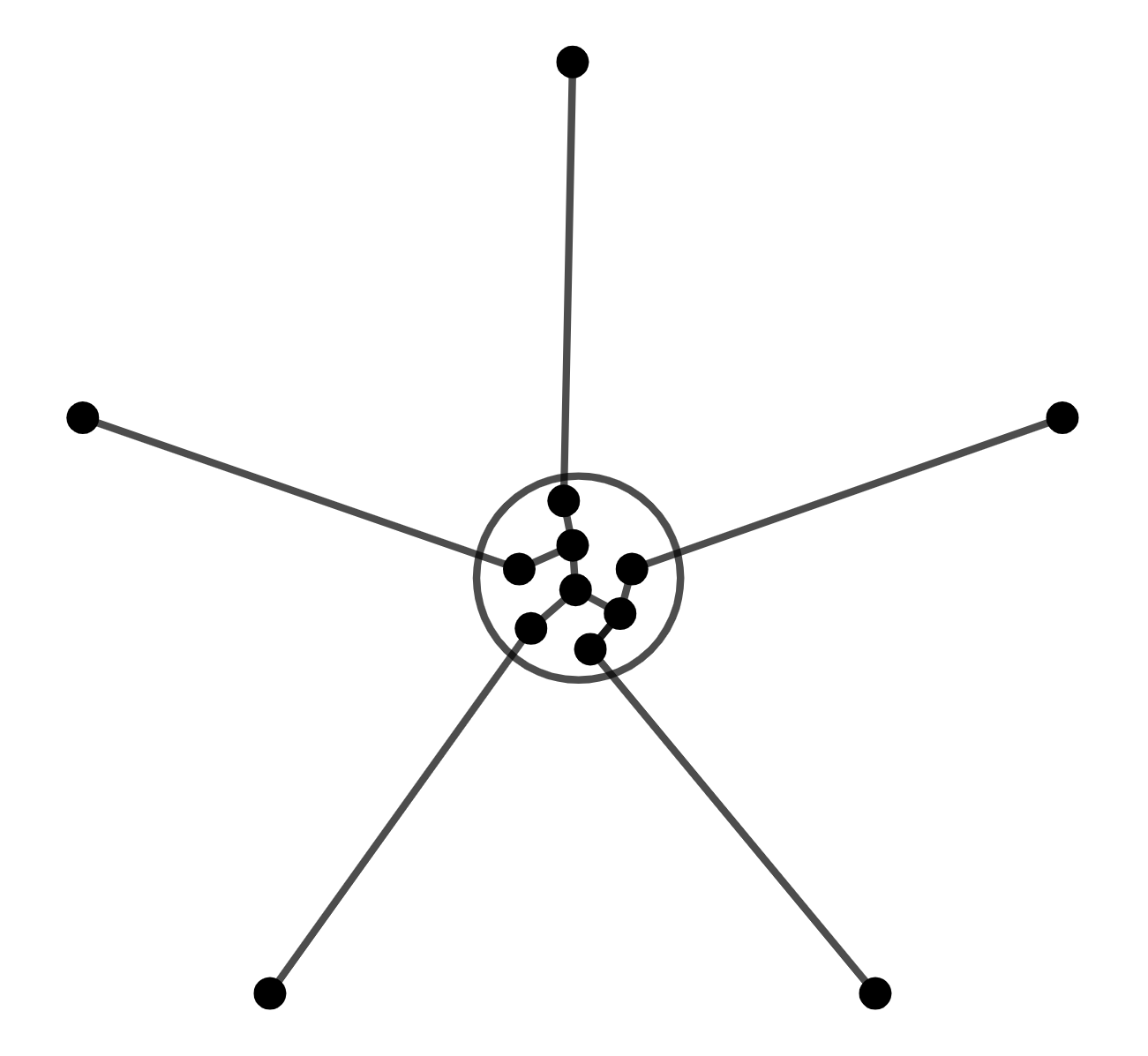}};
            
            \node at (6,-2.6){$(b)$};
        \end{tikzpicture}
        \caption{(a) A point set whose average EMST-ratio approaches approximately $2.154$.\\
        (b) A point set with average EMST-ratio $<1$. 
        The point set consists of $n-5$ points placed in a small circle (the core) and $5$ outer points as vertices of a regular pentagon inscribed in the unit circle around the origin. The corresponding MST is depicted.}
    \label{fig:averag/averagsmall}
\end{figure}
Certainly, Theorem~\ref{thm:metricupperbound} implies that $3$ is an upper bound for the average EMST-ratio of point sets in Euclidean space. 
How much can this upper bound be improved?
There exist point sets in $\Rspace^2$ with an average EMST-ratio $ \approx 2.154$.
For instance, consider the unit circle centered at the origin in the plane, an equilateral triangle $abc$ on it, and then a point set $P=A\cup B \cup C \cup \{0\}$, where $A$ (resp. $B$, $C$) consists of many points that lie within distance $\varepsilon$ of $a$ (resp. $b$, $c$), for $\varepsilon$ being small enough (see Figure~\ref{fig:averag/averagsmall}.a). In this case, the average EMST-ratio approaches $\frac{3+2\sqrt{3}}{3} \approx 2.154$ as the number of points increases, because with high~probability, in a random bi-partition, each of the sets $A$, $B$, and $C$ will contain both colors. Regarding the lower bound, we show that the average EMST-ratio for $n$ points is~always at least $\frac{n-2}{n-1}$. This is a stronger result than the first part of Theorem~2 in \cite{DPT23}.

\begin{theorem}\label{thm:averagemst}
    For every set of $n$ points in Euclidean space, the average EMST-ratio is at least $\frac{n-2}{n-1}$.
\end{theorem}

\begin{proof}
    Consider a set of $n$ points $P$ and let $e_1,...,e_{n-1}$ be the edges of an EMST of $P$, called $T$. Assume that the edge $e_i$ has weight $w_i$ and $w_1\le w_2\le \ldots \le w_{n-1}$.
    Considering a coloring $P=R \cup B$ of $P$ into red and blue points, let $i$ be the~smallest index such that the edge $e_i$ is colorful. As $\text{EMST}(R) \cup \text{EMST}(B) \cup \{e_i\}$ is a spanning tree of $P$, it holds that $|\text{EMST}(R)|+|\text{EMST}(B)|\ge|\text{EMST}(P)|-w_i$.

    On the other hand, considering a random coloring of $P$ with two colors, the probability that the edge $e_i$ is colorful is $\frac{1}{2}$, and since the $e_i$'s do not form a cycle, their corresponding events are independent. Thus, for a fixed index $j$, the probability that in a random coloring, the shortest colorful edge of $T$ is $e_j$ is $(\frac{1}{2})^j$. Therefore, there are precisely $2^{n-j}$ colorings of $P$ in which $j$ is the smallest index such that the edge $e_j$ is colorful.
    By the above argument and since there are $2^n-2$ possible (proper) colorings, the average EMST-ratio $a$ is at least:
    \begin{equation*}
    a \geq |\text{EMST}(P)|-\sum_{i=1}^{n-1}\frac{2^{n-i}w_i}{2^n-2}.    \end{equation*}
    As we have $\sum_{i=1}^{n-1}2^{n-i}=2^n-2$ and $w_1\le w_2\le \ldots \le w_{n-1}$, it can be concluded that $a\ge|\text{EMST}(P)|-\frac{1}{n-1}\cdot|\text{EMST}(P)|$,
    and hence the claim.  
    \qedhere
\end{proof}

Note that Theorem~\ref{thm:averagemst} holds also for general graphs.
The average EMST-ratio appears to be greater than 1 in many cases, though this is not always true. The only instances we found in $\Rspace^2$ consist of point sets where $n-6 \leq k \leq n-1$ points are clustered closely together, forming a core, and $1 \leq n-k \leq 6$ points are sufficiently far from the core such that their pairwise distances are not less than their distances to the core. We believe these are the only point sets with an average EMST-ratio of less than $1$. As an example, in Section~\ref{apndx:avg<1}, we show that the average EMST-ratio for the point set in Figure~\ref{fig:averag/averagsmall}.b is less than $1$.



\begin{question}
    Which point sets in $\Rspace^2$ have an average EMST-ratio $< 1$?
\end{question}

\section{Discussion}\label{sec:discussion}
In addition to the conjectures and questions posed in this paper, there are several other directions that are worth exploring further.
\paragraph*{Coloring with Many Crossings.}
For a point set $P$ in the plane, partitioned as $P = R \cup B$ of $P$, the \textit{Chromatic Crossing Number} is the number of crossings between pairs of edges, one in $\text{EMST}(R)$ and one in $\text{EMST}(B)$. Intuitively, this also measures how mixed the coloring is; the higher the number, the more mingled it appears. In fact, for any point set, it is easy to find colorings with no crossings between the two MSTs. More interesting, however, is the \textit{Maximum Chromatic Crossing Number}, which represents the highest chromatic crossing number achievable across all possible colorings of the point set.

\begin{figure}[ht]
  \centering
  \includegraphics[width=0.25\linewidth]{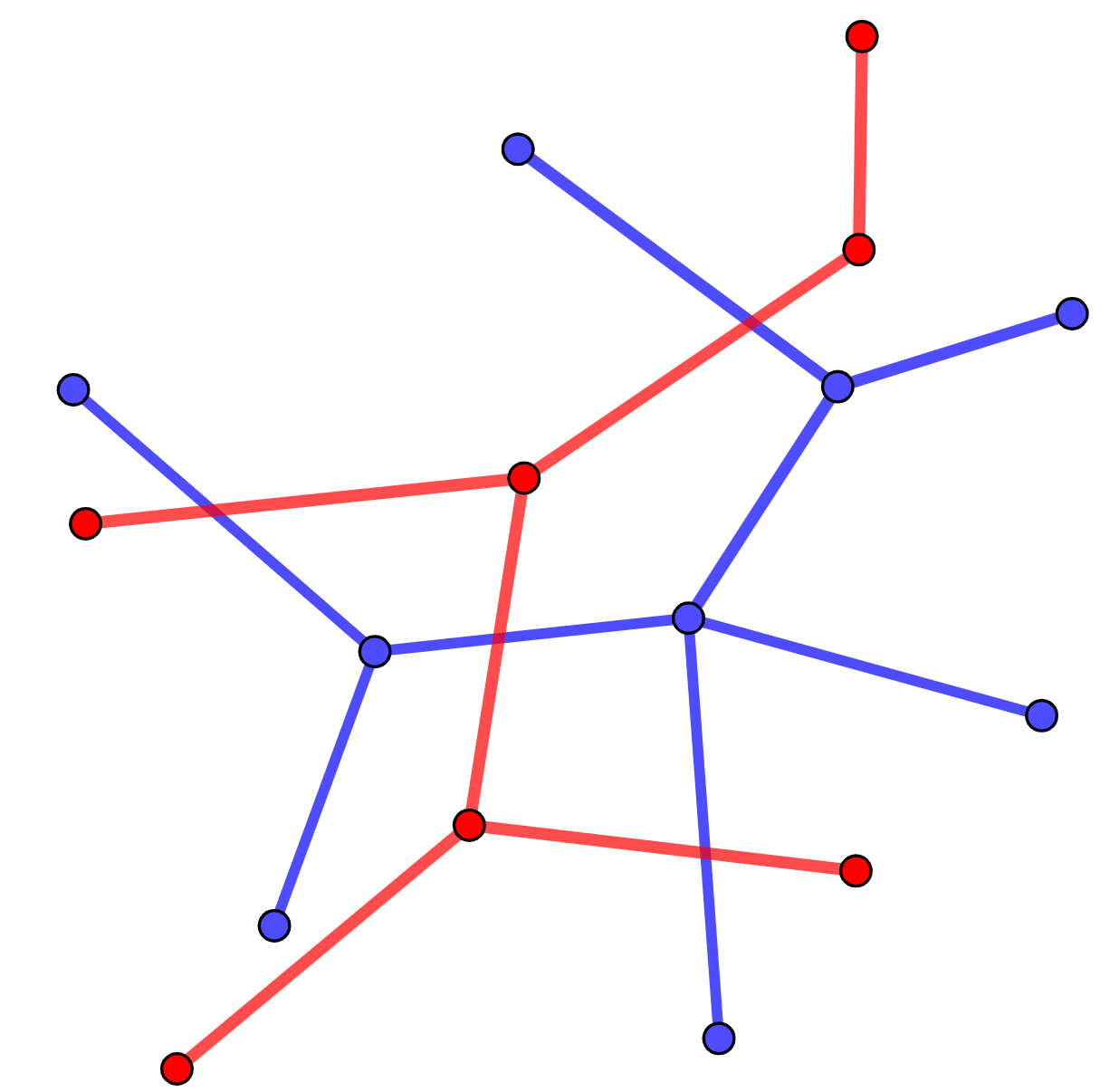}
  \caption{A point set and a bi-partition with chromatic crossing number $4$.}
  \label{fig:crossing}
\end{figure}
\begin{question}
    How large and how small can the maximum chromatic crossing number of a set of $n$ points in the plane be?
\end{question}

\begin{question}
   Are there relations between the chromatic crossing number and the EMST-ratio? In particular, does the coloring that maximizes the chromatic crossing number tend to have a relatively high EMST-ratio?
\end{question}

\begin{conjecture}
    Finding the maximum chromatic crossing number of a point set in the plane is NP-hard.
\end{conjecture}

\paragraph*{Maximum MST Length over Subsets.} A related question to finding the maximum MST-ratio is the \textit{MAX-EMST-subset} problem. Here, given a point set $P$, we aim to find the maximum EMST length over all subsets $Q \subseteq P$. This extends to the abstract setting as the \textit{MAX-MST-subset} problem, where, given a weighted graph $G$, we seek the maximum MST weight over all induced subgraphs $H \subseteq G$.

\begin{question}
    Are MAX-MST-subset and MAX-EMST-subset NP-hard?
\end{question}
\paragraph*{More Colors.}
Similar questions can be asked for point sets with $c>2$ colors. To express the MST-ratio, we consider the total length of MST of points in each color divided by the length of an MST of the (uncolored) point set. Theorem \ref{thm:metricupperbound} can be generalized to obtain $2c-1$ as an upper bound for the MST-ratio in any metric space. For points in Euclidean space, this bound can be improved for $c\ge 4$. In particular, as $\rho_d\ge 0.577$ for every $d$, the maximum EMST-ratio is upper bounded by $\frac{c}{0.577}$. It remains unclear how the results in \cite{CDES24} on lattices can be generalized to more than $2$ colors. Furthermore, exploring the average EMST-ratio for $c$ colors is an interesting direction to pursue.
%
%
%
%

\appendix
\section{Appendix}
\subsection{Maximum EMST-ratio of Triangular Chain}\label{subsec:zigzagproof}
In this section, we show that the maximum EMST-ratio for the triangular chain, shown in Figure~\ref{fig:BipartitevsMax} is $\frac{3k-2}{2k+1}$. Let $P$ be the point set of a triangular chain with $2k+1$ points and let us number the points of $P$ from left to right by $p_1$, $p_2$, $\ldots$, $p_{2k+1}$. 

As $|\text{EMST}(P)|=2k$, to prove the claim we show that for any proper coloring of the chain to red and blue, we have $|\text{EMST}(R)|+|\text{EMST}(B)|\le 3k-2$.
We prove this by applying induction on $k$. The claim holds trivially for $k=1$, as any proper coloring has $|\text{EMST}(R)|+|\text{EMST}(B)|=1=3-2$.

Now, assume that the claim holds for $k$. Now we prove the claim for $k+1$. Assume we have a chain with $2(k+1)+1$ points. Let $B$ and $R$ be the set of blue points and red points in a proper coloring of $P$. We distinguish two cases:

\textbf{a) If there exists $i\in \{1,3\}$, such that $p_i$, $p_{i+1}$ and $p_{i+2}$ have the same color:} W.l.o.g. we assume these points belong to $R$. 
First, consider the case $i=1$. In this case, $R\setminus \{p_1,p_2\}$ and $B$ is a proper coloring for $P\setminus \{p_1,p_2\}$. Also, $T_R:=\text{EMST}(R\setminus \{p_1,p_2\})\cup\{p_1p_2,p_2p_3\}$ is a spanning tree of $R$. By induction hypothesis $|\text{EMST}(R\setminus \{p_1,p_2\})|+|\text{EMST}(B)|\le 3k-2$, which implies that:
$$|\text{EMST}(R)|+|\text{EMST}(B)|\le |T_R|+|\text{EMST}(B)|$$ $$\le3k-2+|p_1p_2|+|p_2p_3|=3k-1+2<3(k+1)-2,$$
and hence the claim.

Now consider the case $i\neq 1$ (i.e. $i=3$). In this case, let $P'$ be the triangular chain obtained from $P$ by removing $p_1$ and $p_2$. Let $R'$ and $B'$ be a proper coloring of $P'$, where:
\begin{itemize}
    \item For every $j\ge 5$, $p_j\in R'$ iff $p_j\in R$.
    \item For every $j\in \{3,4\}$, $p_j\in R'$ iff $p_{j-2}\in R$.
\end{itemize}

Note that $R'$ and $B'$ is a proper coloring of $P'$ and hence by induction hypothesis:
$$|\text{EMST}(R')|+|\text{EMST}(B')|\le 3k-2.$$
As $i\neq 1$, then $(p_1\cup p_2)\cap B\neq \emptyset.$ We have $|\text{EMST}(R')|=|\text{EMST}(R)|-2$ and $|\text{EMST}(B')|+1\ge|\text{EMST}(B)|$. Therefore: 
$$|\text{EMST}(R)+|\text{EMST}(B)|\le |\text{EMST}(R')|+|\text{EMST}(B')|+3\le 3(k+1)-2.$$

\textbf{b) Otherwise:} W.l.o.g. assume $p_3\in R$. 
First, assume $p_1$ and $p_2$ are in $B$. Since we are not in case a), then at least one of the points $p_4$ and $p_5$ is in $B$. Therefore, $R$ and $B\setminus \{p_1,p_2\}$ is a proper coloring for $P\setminus \{p_1,p_2\}$. Now $T_B:=\text{EMST}(B\setminus \{p_1,p_2\})\cup\{p_1p_2,p_2p_i\}$ is a spanning tree of $B$, where $p_i\in \{p_4,p_5\}\cap B$.
Note that by the induction hypothesis $|\text{EMST}(B\setminus \{p_1,p_2\})|+|\text{EMST}(R)|\le 3k-2$, therefore
$$|\text{EMST}(R)|+|\text{EMST}(B)|\le |T_B|+|\text{EMST}(R)|$$ $$\le 3k-2+|p_1p_2|+|p_2p_i|=3k-1+|p_2p_i|<3(k+1)-2.$$

Now assume that there exists $j\in \{1,2\}$ such that $p_j\in R$. Since we are not in case a), then $p_{3-j}$ is blue and at least one of the points $p_4$ and $p_5$ (say $p_i$) is in $B$. Therefore, $R$ and $B\setminus \{p_1,p_2\}$ is a proper coloring for $P\setminus \{p_1,p_2\}$. Now $T_B:=\text{EMST}(B\setminus \{p_1,p_2\})\cup\{p_jp_3,p_{3-j}p_i\}$ is a spanning tree of $B$.
Note that by induction hypothesis $|\text{EMST}(B\setminus \{p_1,p_2\})|+|\text{EMST}(R)|\le 3k-2$, therefore
$$|\text{EMST}(R)|+|\text{EMST}(B)|\le |T_B|+|\text{EMST}(R)|$$ $$\le3k-2+|p_jp_3|+|p_{3-j}p_i|=3k-1+|p_{3-j}p_i|\le 3(k+1)-2.$$



\subsection{A point set with average EMST-ratio $<1$}
\label{apndx:avg<1}
In this section, we show that the average EMST-ratio for the point set described in Figure~\ref{fig:averag/averagsmall}.b is smaller than $1$.
The point set consists of $n-5$ points in a small circle through the origin whose radius is $\varepsilon \in {\mm{o}}(\frac{1}{n2^n})$ (the core), along with $5$ outer points as vertices of a regular pentagon inscribed in the unit circle around the origin. Observe that each side of this pentagon is $s=2\sin{36}\approx1.176$ and each diagonal is $d=\frac{1+\sqrt{5}}{2}s\approx1.902$.
In this example, clearly $|\text{EMST}(P)|\ge 5- {O}(n\varepsilon)$. In any bi-coloring of $P$ by red and blue if an outer point has the same color as one of the core points, then the edge connecting this outer point to a core point appears in one of the monochromatic EMSTs.
Therefore, in $2^n-64$ bi-colorings of $P$ where the core contains both colors, we have $|\text{EMST}(R)|+|\text{EMST}(B)|\le 5+{O}(n\varepsilon)$. There are only $2^6-2=62$ proper bi-colorings with a mono-chromatic core. Consider one of these $62$ colorings and assume the points in the core are red. In this case, if $0\le n_r<5$ denotes the number of outer points that are red, then $|\text{EMST}(R)|=n_r+O(n\varepsilon)$. Also, if the outer blue points are consecutive vertices of the outer pentagon, $|\text{EMST}(B)|=(4-n_r)s$ and otherwise $|\text{EMST}(B)|=(3-n_r)s+d$. So in each of these colorings $|\text{EMST}(R)|+|\text{EMST}(B)|\le d+s+2+{O}(n\varepsilon)<5.1$. Moreover, in $62-2\times5=52$ of such colorings we even have $|\text{EMST}(R)|+|\text{EMST}(B)|\le d+3+{O}(n\varepsilon)<4.91$. Thus, the average EMST-ratio for this point set is upper bounded by: 
\begin{equation*}\frac{(2^n-64)\cdot (5+ {O}(n\varepsilon))+52\cdot (4.91)+10\cdot (5.1)}{(2^n-2)\cdot (5- {O}(n\varepsilon))}<1.
\end{equation*}

\subsection{Computational Experiments}\label{sec:experiments}
This section presents the results of computational experiments comparing the average EMST-ratio, maximum EMST-ratio, and the EMST-ratio for a bipartite coloring, all conducted on random point sets. The setup is as follows. In a single experiment, we consider $n$ points selected uniformly at random in the unit square, compute the EMST-ratios for all $2^n-2$ proper colorings of this point set, and take the maximum and the average over all colorings. Additionally, we compute the EMST-ratio for a bipartite coloring of the sample point set, by computing an MST using Prim's algorithm \cite{KS07}. For $n$ between $5$ and $20$, we calculate the mean values of the functions (i.e. the average EMST-ratio, maximum EMST-ratio, and bipartite EMST-ratio) across $500$ experiments. The resulting averages are displayed in Figure~\ref{fig:threecurves}.
According to Theorem~\ref{thm:randomaverage}, the average EMST-ratio (the red curve) tends to $\sqrt{2}$. Interestingly, the diagram in Figure~\ref{fig:threecurves} suggests that a similar behavior might hold for the maximum EMST-ratio (the green curve) as well, though approaching a different limit. In fact, the program did not find any instance of a random point set with a maximum EMST-ratio greater than $2$ in our experiments.

\begin{figure}[h!]
  \centering
  \includegraphics[width=0.95\linewidth]{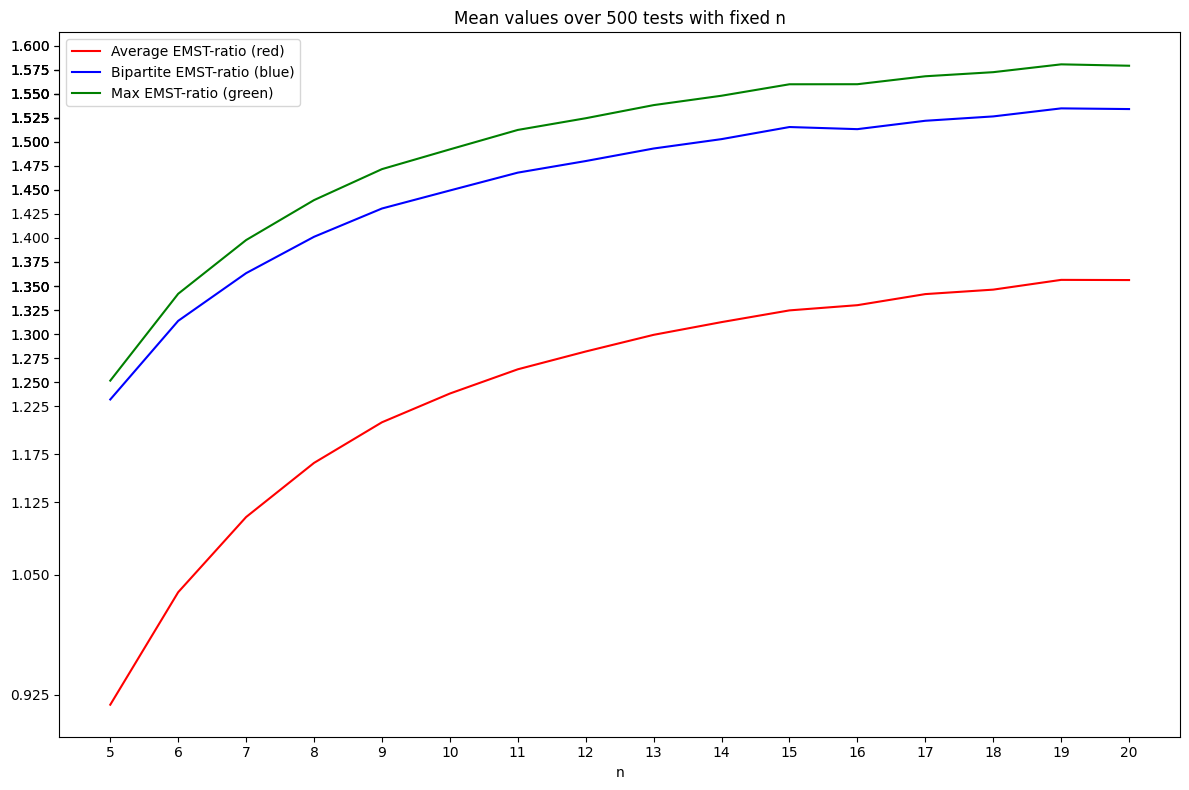}
  \caption{The mean values of the maximum EMST-ratios (green), the bipartite EMST-ratios (blue), and the average EMST-ratios (red) of $500$ random sets of $n$ points, where $5 \le n \le 20$.}
  \label{fig:threecurves}
\end{figure}

\begin{question}
    What is the limit for the expected maximum EMST-ratio of $n$ points sampled uniformly at random in $[0,1]^2$, as $n$ goes to infinity?
\end{question}

\begin{conjecture}
For $n$ random points uniformly distributed in $[0,1]^2$, the maximum EMST-ratio is less than $2$ with probability tending to $1$ as $n \rightarrow{\infty}$.
\end{conjecture}

As for the bipartite coloring, Figure~\ref{fig:threecurves} clearly shows how close its EMST-ratio is to the maximum EMST-ratio in expectation. Figure~\ref{fig:bipmaxdiagram} particularly compares the EMST-ratio for a bipartite coloring computed using Prim's algorithm and the maximum EMST-ratio across $10^5$ experiments involving $5$ to $20$ points. 
Moreover, the maximum EMST-ratio is never more than $1.3$ times and is rarely more than $1.1$ times the computed bipartite EMST-ratio across all these experiments.

\begin{figure}[ht]
  \centering
  \includegraphics[width=0.5\linewidth]{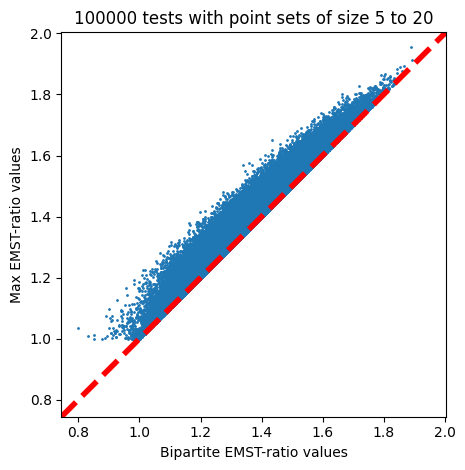}
    \caption{Each point in the diagram represents a random point cloud, with the $x$ and $y$ coordinates corresponding to a bipartite EMST-ratio of the sample point set computed using Prim's algorithm and the maximum EMST-ratio, respectively.}  \label{fig:bipmaxdiagram}
\end{figure}

\begin{conjecture}
    For $n$ random points uniformly distributed in $[0,1]^2$, the maximum EMST-ratio is less than $1.1$ times any bipartite EMST-ratio of the point set with probability tending to $1$ as $n$ goes to infinity.
\end{conjecture}

\begin{figure}[ht]
  \centering
  \includegraphics[width=0.5\linewidth]{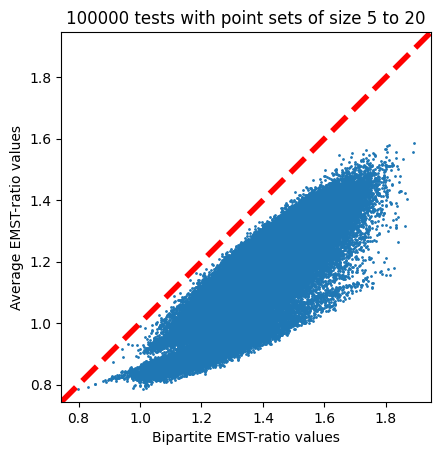}
    \caption{In this diagram each point represents a random point cloud, with the $x$ and $y$ coordinates corresponding to a bipartite EMST-ratio of the sample point set computed using Prim's algorithm and the average EMST-ratio, respectively.}
 \label{fig:avebipdiagram}
\end{figure}
Figure~\ref{fig:avebipdiagram}, on the other hand, presents a comparison between a bipartite EMST-ratio computed using Prim's algorithm and the average EMST-ratio obtained from $10^5$ experiments with point sets ranging from $5$ to $20$ points. 
Interestingly, in all computed experiments without exception, the computed bipartite EMST-ratio exceeds the average EMST-ratio. This suggests that bipartite colorings of a point set provide better EMST-ratios than what is expected from a random coloring. Although this is not always the case for finite point sets, we believe that there are very few counterexamples.

\begin{question}
    Can the point sets in $\Rspace^2$ for which a bipartite EMST-ratio of them is smaller than the average EMST-ratio be characterized?
\end{question}
\end{document}

%% file: Figures/SteinerRatioMax.tex
\begin{tikzpicture}[scale=0.6]

                \tikzset{black dot/.style={draw=black, very thick, circle,minimum size=0pt, inner sep=1pt, outer sep=1pt,fill=black}}
                \tikzset{terminal/.style={draw=black,  thick,minimum size=0pt, inner sep=2.5pt, outer sep=1pt}}
                \tikzset{P node/.style={fill={rgb,255: red,20; green,154; blue,0}, draw={rgb,255: red,20; green,154; blue,0}, circle, minimum size=0pt,inner sep=1pt, outer sep=1pt}}
            
                \tikzstyle{witness edge}=[-, draw={rgb,255: red,195; green,0; blue,3}, very thick]
                \tikzstyle{T edges}=[-, very thick]
                \tikzstyle{new witness}=[-, draw={rgb,255: red,195; green,0; blue,3}, dashed, very thick]
                \tikzstyle{connected terminals}=[-, draw=black, dashed, very thick]
                \tikzstyle{P}=[-, draw={rgb,255: red,20; green,154; blue,0}, very thick]

            \node [style=terminal] (u) at (-6, 7.5) {};
            \node [style=terminal] (v) at (-4, 4) {};
            \node [style=terminal] (w) at (-8, 4) {};
            \node [style=black dot] (c) at (-6, 5.3) {};

            \draw [style=P] (u) to (c);
            \draw [style=P] (v) to (c);            
            \draw [style=P] (w) to (c);
            \draw  [style=T edges] (w) to (v);
            \draw  [style=T edges] (w) to (u);
\end{tikzpicture}

%% file: Figures/DFS.tex
\begin{tikzpicture}[scale=0.6]

                \tikzset{black dot/.style={draw=black, very thick, circle,minimum size=0pt, inner sep=1pt, outer sep=1pt,fill=black}}
                \tikzset{blue dot/.style={draw=blue, very thick, circle,minimum size=0pt, inner sep=1pt, outer sep=1pt,fill=blue}}
                \tikzset{red dot/.style={draw=red, very thick, circle,minimum size=0pt, inner sep=1pt, outer sep=1pt,fill=red}}
                \tikzset{terminal/.style={draw=black,  thick,minimum size=0pt, inner sep=2.5pt, outer sep=1pt}}
                \tikzset{P node/.style={fill={rgb,255: red,20; green,154; blue,0}, draw={rgb,255: red,20; green,154; blue,0}, circle, minimum size=0pt,inner sep=1pt, outer sep=1pt}}
            
                \tikzstyle{witness edge}=[-, draw={rgb,255: red,195; green,0; blue,3}, very thick]
                \tikzstyle{T edges}=[-, very thick]
                \tikzstyle{new witness}=[-, draw={rgb,255: red,195; green,0; blue,3}, dashed, very thick]
                \tikzstyle{connected terminals}=[-, draw=black, dashed, very thick]
                \tikzstyle{P}=[-, draw={rgb,255: red,20; green,154; blue,0}, very thick]

                




                \node [style=red dot] (r) at (8,0){};
                \node [style=blue dot] (v1) at (6,-1){};
                \node [style=blue dot] (v2) at (8,-1){};
                \node [style=red dot] (v3) at (10,-1){};
                \node [style=blue dot] (u11) at (5,-2){}; 
                \node [style=red dot] (u111) at (4,-3){};
                \node [style=red dot] (u112) at (6,-3){};
                \node [style=red dot] (u12) at (7,-2){};
                \node [style=red dot] (u122) at (8,-3){};
                
                \draw [black, very thick] (v3) to (r) to (v2);
                \draw [black, very thick] (u111) to (u11) to (u112);
                \draw [black, very thick] (v1) to (u12) to (u122);

                \draw [green!70!black, very thick] (r) to (v1) to (u11);

                \draw [red, very thick, dashed] (u111) to (u112);
                \draw [red, very thick, dashed] (r) to (u12) to (u112);
                \draw [red, very thick, dashed, bend right=30] (v3) to (r);
                \draw [red, very thick, dashed, bend right=30] (u122) to (u12);
                \draw [blue, very thick, dashed, bend right=40] (v1) to (u11);
                \draw [blue, very thick, dashed] (v1) to (v2);

                \node () at (8,0.25) {$r$};

                \node () at (7,-2.25) {$u_3$};
                \node () at (6,-3.25) {$u_2$};
                \node () at (4,-3.25) {$u_1$};
                \node () at (6,-0.75) {$w$};

\end{tikzpicture}

%% file: Maximum_MST_Ratio_LNCS.bbl
\begin{thebibliography}{21}

\footnotesize{
\bibitem{AvBe92}
{\sc F.\ Avram and D.\ Bertsimas.}
The minimum spanning tree constant in geometrical probability and under the independent model: a unified approach.
\emph{Ann.\ Appl.\ Probab.} {\bf 2} 113--130 (1992).

\bibitem{AESW07}
{\sc P. Agarwal, H. Edelsbrunner, O. Schwarzkopf, E. Welzl.}
Euclidean minimum spanning trees and bichromatic closest pairs.
\emph{Discrete Comput. Geom.} V. 6, pp 407–422 (1991).

\bibitem{BCDES22}
{\sc R.\ Biswas, S.\ Cultrera di Montesano, O.\ Draganov, H.\ Edelsbrunner, and M.\ Saghafian.}, On the Size of Chromatic Delaunay Mosaics. \texttt{arXiv.2212.03121} (2022).


\bibitem{BHH59}
{\sc J.\ Beardwood, J.H.\ Halton and J.M.\ Hammersley.}
The shortest path through many points.
\emph{Math.\ Proc.\ Cambridge Phil.\ Soc.} {\bf 55} 299--327 (1959).

\bibitem{BCKO08}
{\sc M. de Berg, O. Cheong, M. van Kerveld, and M. Overmars.}
Computational Geometry:
Algorithms and Applications (Third ed.).
\emph{Springer} (2008).

\bibitem{Ber12}
{\sc S. Bernstein.}
Démonstration du théorème de Weierstrass fondée sur le calcul des probabilités (Proof of the theorem of Weierstrass based on the calculus of probabilities).
\emph{Comm. Kharkov Math. Soc.} 13: 1--2 (1912).

\bibitem{CG85}
{\sc F.R.K.\ Chung and R.L.\ Graham.}
A new bound for Euclidean Steiner minimal trees.
\emph{Discrete Geometry and Convexity}, Annals N.Y.\ Acad.\ Sci. 328--346 (1985).

\bibitem{CDES22}
{\sc S.\ Cultrera di Montesano, O.\ Draganov, H.\ Edelsbrunner, and M.\ Saghafian.} 
Chromatic alpha complexes. \emph{Foundations of Data Science} (2025). 

\bibitem{CDES24}
{\sc S.\ Cultrera di Montesano, O.\ Draganov, H.\ Edelsbrunner, and M.\ Saghafian.} 
The Euclidean MST-ratio for bi-colored Lattices.
\emph{Symp. on Graph Drawing and Network Visualization} (2024).

\bibitem{CLRS}
{T. Cormen, C.E. Leiserson, R.L. Rivest, C. Stein.} Introduction To Algorithms (Third ed.).
\emph{MIT Press} pp. 631. (2009).


\bibitem{DW87}
{\sc B.V. Dekster and J.B. Wilker.}
Edge-lengths guaranteed to form a simplex.
\emph{Arch. Math. 49} 351-366
(1987).

\bibitem{DS24}
{\sc O. Draganov and M. Saghafian.}
Private Communication (2024).

\bibitem{DPT23}
{\sc A.\ Dumitrescu, J.\ Pach and G.\ T\'{o}th.}
Two trees are better than one.
\texttt{arXiv:2312.09916} (2023).



\bibitem{Gil65}
{\sc E.N.\ Gilbert.}
Random minimal trees.
\emph{J.\ Soc.\ Industr.\ Appl.\ Math.} {\bf 13} 376--387 (1965).

\bibitem{GP68}
{\sc E.N.\ Gilbert and H.O.\ Pollack.}
Steiner minimal trees.
\emph{SIAM J.\ Appl.\ Math.} {16} 1--29 (1968).

\bibitem{GH76}
{\sc R.L. \ Graham and F. K. \ Hwang.}
Remarks on Steiner minimal trees.
\emph{Bull. Inst. Math. Acad. Sinica} {\bf 4(1)} 177--182 (1976).

\bibitem{KS07}
{\sc L. Koralov and Y. Sinai.}
Probabilistic proof of the Weierstrass theorem.
\emph{Theory of probability and random processes (2nd ed.). Springer} p. 29 (2007).

\bibitem{PDC14}
{\sc G. Thiyagarajan,  and M. Saravanan.}
Path Double Covering Number of a Graph.
\emph{International Journal Of Mathematics And Computer Research}, 2(08), 565-573 (2014).

\bibitem{Prim57}
{\sc R.C. Prim.}, Shortest connection networks And some generalizations,
\emph{Bell System Technical Journal}, 36 (6): 1389–1401 (1957).


\bibitem{Z06}
{\sc D. Zuckerman.}
Linear degree extractors and the inapproximability of max clique and chromatic number.
\emph{Proceedings of the 38th Ann. ACM Symp. on Theory of Computing} {10} 681--690 (2006).

}


\end{thebibliography}
